\documentclass{llncs}

\usepackage{algpseudocode} %
\usepackage{amsmath,amssymb} %
\usepackage[english]{babel} %
\usepackage{bbm} %
\usepackage{booktabs} %
\usepackage{color} %
\usepackage{colortbl} %
\usepackage{derivative} %
\usepackage{dsfont} %
\usepackage{enumitem} %
\usepackage{float} %
\usepackage{flushend} %
\usepackage[T1]{fontenc} %
\usepackage{framed} %
\usepackage{graphicx} %
\usepackage[colorlinks,citecolor=blue,bookmarks=true,breaklinks=true]{hyperref} %
\usepackage[capitalize]{cleveref} 
\usepackage{hyphenat} %
\usepackage{ifdraft} %
\usepackage[utf8]{inputenc} %
\usepackage{interval} %
\usepackage{mathrsfs} %
\usepackage{mathtools} %
\usepackage{mdframed} %
\usepackage{microtype} %
\usepackage{soul} %
\usepackage{tikz} %
\usepackage{suffix} 
\usepackage{xcolor} %
\usepackage{xparse} %
\usepackage{xspace} %
\usepackage{physics}
\usepackage{multirow}
\usepackage{wrapfig}
\usepackage{cite}

\setul{1ex}{.5pt} %

\intervalconfig{soft open fences}

\newfloat{algorithm}{hbt}{lop} %
\floatname{algorithm}{Algorithm} %
\crefname{algorithm}{Algorithm}{Algorithms} %

\definecolor{White}{rgb}{1,1,1}
\definecolor{Black}{rgb}{0,0,0}
\definecolor{Gray}{rgb}{.3,.3,.3}
\definecolor{LightGray}{rgb}{.8,.8,.8}
\definecolor{LightYellow}{rgb}{0.95,0.90,0.75}
\definecolor{Yellow}{rgb}{0.95,0.8,0.45}
\definecolor{LightBlue}{rgb}{0.85,0.9,.95}
\definecolor{Blue}{rgb}{0.25,0.45,.75}
\definecolor{LightRed}{rgb}{0.95,0.9,0.85}
\definecolor{Red}{rgb}{0.85,0.25,0.25}
\definecolor{LightGreen}{rgb}{0.75,0.95,0.85}
\definecolor{Green}{rgb}{0.25,0.75,.45}

\colorlet{ChannelColor}{LightYellow}
\colorlet{ChannelTextColor}{Black}
\colorlet{ReadoutColor}{White}
\colorlet{LineColor}{Black}
\colorlet{TextColor}{Black}

\usetikzlibrary{positioning,calc,decorations.pathreplacing}

\tikzset{%
  draw=LineColor, text=LineColor, %
  control/.style={circle, fill=LineColor, minimum size=3, inner sep=0}, %
  gate/.style={draw, minimum size=14, fill=ChannelColor, inner sep=2,
    text=ChannelTextColor}, %
  target/.style={circle, draw, minimum size=6, inner sep=0}, %
  cross/.style={cross out, draw, minimum size=2.5, inner sep=0}, %
}

\newcommand{\microspace}{\mspace{.5mu}} %
\renewcommand{\ket}[1]{\ensuremath{\lvert\microspace#1%
    \microspace\rangle}} %
\renewcommand{\bra}[1]{\ensuremath{\langle\microspace#1%
    \microspace\rvert}} %

\newcommand{\ignore}[1]{} %

\newcommand{\complex}{\mathbb{C}} %

\newcommand{\class}[1]{\textup{#1}\xspace} %

\newcommand{\NP}{\class{NP}} %
\newcommand{\BQP}{\class{BQP}} %
\newcommand{\MIP}{\class{MIP}} %
\WithSuffix\newcommand\MIP*{\ensuremath{\class{MIP}^*}} %

\DeclareMathOperator{\poly}{poly} %

\def\e{e} 

\newcommand\restr[2]{{
  \left.\kern-\nulldelimiterspace
  #1 
  \vphantom{\big|} 
  \right|_{#2} 
  }}

\newcommand{\size}[1]{\ensuremath{\mathrm{B}(#1)}} %
\newcommand{\bin}{\ensuremath{\{0, 1\}}} %
\newcommand{\lin}{\mathrm{in}} %
\newcommand{\lout}{\mathrm{out}} %

\newcommand{\gatestyle}[1]{\mathrm{#1}} %
\newcommand{\Had}{\gatestyle{H}} %
\newcommand{\T}{\gatestyle{T}} %
\newcommand{\X}{\gatestyle{X}} %
\newcommand{\Y}{\gatestyle{Y}} %
\newcommand{\Z}{\gatestyle{Z}} %
\newcommand{\W}{\gatestyle{W}} %
\newcommand{\Id}{\gatestyle{I}} %
\newcommand{\CZ}{\gatestyle{CZ}} %
\newcommand{\CCZ}{\gatestyle{CCZ}} %
\newcommand{\CNOT}{\gatestyle{CNOT}} %
\newcommand{\Toffoli}{\gatestyle{Toffoli}} %
\newcommand{\iSWAP}{\gatestyle{iSWAP}} %
\newcommand{\fSim}{\gatestyle{fSim}} %
\newcommand{\Rz}{\gatestyle{Rz}}

\DeclareMathOperator{\cc}{cc} 
\DeclareMathOperator{\tw}{tw} 

\begin{document}


\title{FeynmanDD\@: Quantum Circuit Analysis with Classical Decision Diagrams}

\newcommand*\samethanks[1][\value{footnote}]{\footnotemark[#1]}

\author{
  Ziyuan Wang\thanks{First authors.}\inst{1} \and %
  Bin Cheng\samethanks\inst{2} \and %
  Longxiang Yuan\inst{1} \and %
  Zhengfeng Ji\thanks{Correspondence author. Email: jizhengfeng@tsinghua.edu.cn.}\inst{1,3}}

\institute{Department of Computer Science and Technology, Tsinghua University
\and
Centre for Quantum Technologies, National University of Singapore
\and
Zhongguancun Laboratory}

\date{\today}

\maketitle

\begin{abstract}

  Applications of decision diagrams in quantum circuit analysis have been an
  active research area.
  Our work introduces FeynmanDD, a new method utilizing standard and
  multi-terminal decision diagrams for quantum circuit simulation and
  equivalence checking.
  Unlike previous approaches that exploit patterns in quantum states and
  operators, our method explores useful structures in the path integral
  formulation, essentially transforming the analysis into a counting problem.
  The method then employs efficient counting algorithms using decision diagrams
  as its underlying computational engine.
  Through comprehensive theoretical analysis and numerical experiments, we
  demonstrate FeynmanDD's capabilities and limitations in quantum circuit
  analysis, highlighting the value of this new BDD-based approach.


\end{abstract}

\section{Introduction}\label{sec:intro}


Binary Decision Diagrams (BDDs) are widely regarded as one of the most
influential data structures for representing, analyzing, and simulating
classical circuits~\cite{Bry95,MMBS04}.
Since their introduction in the 1980s~\cite{Bry86}, BDDs have found applications
in various domains, including circuit synthesis, formal verification, and model
checking.
In one of his video lectures given in the 2000s on trees and BDDs, Donald Knuth
described them as ``one of the few truly fundamental data structures that came
out in the last twenty-five years''.
In his seminal book series, \emph{The Art of Computer Programming}, he dedicated
an entire section to BDDs~\cite{Knu09}.


Given their success in classical circuit design and analysis, it is natural to
explore how BDDs can be extended to the representation and simulation of quantum
circuits.
The core idea of decision diagrams is to exploit repeated patterns in truth
tables, enabling compact data representations.
Quantum states and operations, like truth tables, scale exponentially yet may
also exhibit repetitive patterns.
For example, directly representing an $n$-qubit state requires storing $2^{n}$
complex numbers, known as amplitudes.
Several existing works in quantum computing adopt this approach to identify
patterns in quantum states and
operations~\cite{VMH03,AP06,VMH07,NWM+16,LWK11,ZW19,ZHW19,BBW21,WHB22,HZK+22,VGH+23,JFB+24}.
Some studies, like those in~\cite{HZL+22,LOLS24}, combine tensor network methods
with decision diagrams to enhance computational efficiency.
Others, such as~\cite{SCR24}, incorporate additional compression techniques
based on context-free languages.
These methods, inspired by BDDs, extend beyond traditional BDDs and are better
described as \emph{BDD-motivated} techniques.


In BDD-motivated approaches, entirely new implementations of the data structure
are often required.
Additionally, to capture more intricate patterns, non-trivial labels may need to
be assigned to the diagram's links~\cite{NWM+16,ZW19,ZHW19,HZL+22}.
In contrast, a less-explored approach uses classical BDD data structures and
decision diagram packages, such as CUDD~\cite{Som05}, to directly represent and
manipulate quantum states.
For example, Ref.~\cite{TJJ21} proposed a bit-slicing method to represent
quantum amplitudes of states generated by the Clifford and $\T$ gate set using
BDDs.
We refer to this as a \emph{BDD-based} method, as it utilizes classical BDDs as
the underlying engine.
Given the limited research on BDD-based methods for quantum circuits, a natural
question arises: \emph{Are there other BDD-based methods that can significantly
  enhance the analysis and simulation of quantum circuits?}


\textbf{Key Contributions.}
We answer this question in the affirmative by proposing a new BDD-based method.

The first key contribution is the novel integration of the counting capabilities
of decision diagrams in quantum circuit analysis.
The profound connection between counting and quantum
computing~\cite{AA13,FGHP99,Wat09} is well-established, with counting complexity
classes providing natural upper bounds for $\BQP$~\cite{Wat09} and underpinning
quantum supremacy schemes~\cite{AA13}.
While the counting perspective has been explored in quantum circuit
analysis~\cite{MBL24}, this work uniquely combines quantum circuits' counting
nature with BDDs' efficient counting algorithms.
When decision diagrams have bounded size, counting solutions becomes
efficient---a key advantage of BDDs highlighted in Knuth's book~\cite{Knu09}.
This counting algorithm then serves as the computational engine in our method.
Unlike traditional Schrödinger-style simulators that use specialized BDD
variants for state vector evolution~\cite{NWM+16,ZW19,ZHW19,TJJ21,HZL+22}, our
method explores structures in the Feynman-type exponential sums.
To emphasize its foundations in Feynman path integral, we name our method
\emph{FeynmanDD} (Feynman Decision Diagrams).

The second contribution encompasses efficiency-enhancing techniques.
We present a binary synthesis method for constructing FeynmanDD's underlying
decision diagram, carefully transforming low-degree multilinear polynomials to
BDDs through strategic term ordering.
This approach reduces both intermediate representation sizes and computational
complexity.
We also emphasize the importance of variable ordering for BDDs, providing
various ordering heuristics for different circuit families.
The combined consideration of term orders and variable orders proves crucial for
the overall efficiency of our method.
Furthermore, we introduce a sum-of-powers framework that is both
flexible---supporting multiple gate sets within a unified approach---and
efficient, as it delays conversion to BDDs, further optimizing performance.

Finally, extensive numerical experiments demonstrate the superior performance of
FeynmanDD compared to existing decision diagram tools like DDSIM and SliQSim in
amplitude computation while showing advantageous performance in sampling tasks
for many families of circuit types even though our sampling efficiency does not
match that for amplitude estimation.
FeynmanDD also proves effective for circuit equivalence checking for certain
families of circuits, with performance comparable to its simulation
capabilities.
These results confirm that FeynmanDD represents a promising new direction in
quantum circuit analysis that is worthy of further investigation.
In a follow-up work, we will provide a characterization of the complexity of the
FeynmanDD method, showing provable efficiency advantages over tensor network
methods for certain families of circuits.


\textbf{Outline.}
We develop a sum-of-powers (SOP) framework to handle diverse quantum gate sets
in \cref{sec:sop-circ}, consistently mapping circuits of different discrete
quantum gate sets into SOP forms.
As special cases of tensor networks, SOPs can be manipulated and simplified
using tensor techniques, which we explain in \cref{sec:sop-tensor}.
Converting a quantum circuit to its SOP representation is straightforward.
The challenge lies in representing the SOP function $f$ as a BDD (or
multi-terminal BDD~\cite{BFG+93}), which is carefully discussed in
\cref{sec:feynmandd}.
After constructing the BDD for $f_{C}$, we demonstrate in \cref{sec:app} how
FeynmanDD enables quantum circuit analysis and simulation, including amplitude
computation, measurement outcome sampling, and circuit equivalence checking.
To initiate a rigorous understanding of FeynmanDD's capabilities, we construct
circuits using the linear network construction in~\cite{Knu09} (Figure 23),
which FeynmanDD can simulate efficiently, while tensor network and Clifford
methods face provably high complexity.
Extensive numerical experiments on quantum circuit simulation and equivalence
checking are conducted and discussed in \cref{sec:exp} and
\cref{sec:equivalence_check_results} respectively.

\section{Preliminaries}

\subsection{Basics of Quantum Circuits}

This work focuses on pure quantum states and does not consider noise and mixed
states.
An $n$-qubit pure quantum state is a normalized vector in the Hilbert space
$\complex^{2^{n}}$, denoted as $\ket{\psi} = \sum_{x \in \bin^n} c_x \ket{x}$,
where ${ \ket{x} }$ represents the computational basis.
Quantum circuits describe the evolution of quantum states through sequences of
unitary operators, typically acting on one, two, or three qubits.
Important quantum gates include the Hadamard gate $\Had = \frac{1}{\sqrt{2}}
\begin{pmatrix} 1 & \phantom{-}1 \\ 1 & -1 \end{pmatrix}$, the Pauli-$\Z$ gate
$\Z = \begin{pmatrix} 1 & \phantom{-}0 \\ 0 & -1 \end{pmatrix}$, the controlled-$\Z$ gate
($\CZ$), and the controlled-controlled-$\Z$ gate ($\CCZ$).
These gates form a weak universal set capable of approximating any (real)
unitary operator to arbitrary precision and sufficient to simulate all quantum
computation.
Other gates are also considered in this work and are discussed in
\cref{sec:sop-circ}.
Our analysis considers quantum circuits acting on initial state $\ket{0^n}$
without loss of generality.
With $U$ denoting the circuit's unitary operator, the final measurement in the
computational basis yields a classical string $x \in \bin^n$ with probability
$\abs{\bra{x} U \ket{0^{n}}}^2$.

\subsection{Quantum Circuit Analysis}

There are several types of quantum circuit analysis tasks widely considered in
the literature, including circuit simulation, equivalence checking, circuit
synthesis and optimization.
Our method is currently applicable to the simulation and equivalence-checking
tasks.

As our equivalence checking method is reduced to a variant of circuit simulation
in the end, we will focus on the discussion of quantum circuit simulation here.
There are two notions for classical simulation of quantum circuits: the strong
simulation and the weak simulation.
The strong simulation requires computing the output amplitude
$\bra{x}U\ket{0^n}$ (or the output probability) given $x$ and $U$.
The weak simulation requires sampling from the output distribution of the
quantum circuit, i.e., returning a string $x$ with probability
$\abs{\bra{x}U\ket{0^n}}^2$.

A straightforward way to simulate quantum circuits is to use the Schr\"{o}dinger
method, where the quantum state is stored as a $2^n$-dimensional vector~\cite{LWY+20}.
The quantum gates are represented by $2^n \times 2^n$ matrices, and the quantum
state is updated by matrix-vector multiplication.
Since the quantum state is stored, both the strong and weak simulations can be
performed.
The Schr\"{o}dinger method requires exponential space and time complexity, which
quickly becomes infeasible for large quantum circuits.

Another classical simulation technique is based on the Feynman path integral~\cite{FH65}.
Suppose $U$ consists of $m$ gates, $U = U_m\cdots U_2 U_1$.
Then, the idea of Feynman path integral is to insert identity operators in
between, transforming the output amplitude into an exponential sum of products:
\begin{align}\label{eq:feynman-path-integral}
  \bra{x}U\ket{0^n} = \sum_{y_1, \ldots, y_{m-1} \in \bin^n} \bra{x}U_m\ket{y_{m-1}}
  \cdots \bra{y_2}U_2\ket{y_1} \bra{y_1}U_1\ket{0^n}.
\end{align}
Note that each matrix element in the summation can be directly obtained from the
specification of the corresponding gate.
For some universal gate sets, these matrix elements can be represented by a
particularly simple form, giving a sum-of-powers representation as
detailed in \cref{sec:sop-circ}.
This simulation method has polynomial space complexity but the time complexity
is exponential in the \emph{number of gates}, which is even worse than the
Schr\"{o}dinger method.

Instead of directly computing the summation in \cref{eq:feynman-path-integral},
one can also compute it by tensor network contraction, one of the
state-of-the-art classical simulation techniques~\cite{MS08,PZ22}.
A rank-$k$ tensor with bond dimension two can be represented by a
$2^k$-dimensional array, $f_{i_1, \ldots, i_k}$, where
$i_1, \ldots, i_k \in \bin$.
In this way, a single- or two-qubit gate can be represented by a rank-$2$ or
rank-$4$ tensor, respectively.
A tensor network is a collection of tensors, where each index may appear in one
or two tensors.
The contraction of a tensor network is to sum over all indices that shared by
two tensors.
For tensor network representation of the amplitude $\bra{x}U\ket{0^n}$, no open
wires remain and all indices appearing in the tensor network are summed over,
giving a scalar of interest.
In this work, we explore BDD-based techniques for Feynman-type simulation.
We briefly mention that other families of techniques beyond decision diagrams
have also been studied in the literature including, for example, phase
polynomials~\cite{Amy19b,NRS+18}, tensor-based techniques~\cite{MS08,PZ22}, ZX
calculus~\cite{CD08,Wet20,CHKW22}, and tree automata~\cite{CCL+23,ACC+25}.

\subsection{Binary Decision Diagrams}

A binary decision diagram is a rooted, directed acyclic graph that succinctly
represents Boolean functions.
It contains two node types: decision nodes and terminal nodes.
Each decision node carries a variable label, such as $x_i$, with two outgoing
links to child nodes.
The dashed link represents the branch for $x_i = 0$, while the solid link
represents $x_i = 1$.
The terminal nodes, valued as $0$ and $1$, represent the function's output
values.
Any path from the root to a terminal node corresponds to a specific variable
assignment, with the terminal node indicating the Boolean function's value for
that assignment.
An example of a BDD is shown in \cref{fig:bdd}.

Two additional properties of BDDs are important: being ordered and reduced.
An \emph{ordered} BDD ensures that variables appear in a consistent order along
all paths from the root to terminal nodes.
A \emph{reduced} BDD contains no isomorphic subgraphs and no node has identical
children.
In the literature, BDDs typically refer to reduced ordered BDDs (ROBDDs), which
provide a canonical representation of Boolean functions.
This work employs multi-terminal BDDs (MTBDDs), also known as Algebraic Decision
Diagrams (ADDs)~\cite{BFG+93}, a generalized version where terminal nodes can
take multiple values rather than just binary values.

For a Boolean function $f$, we denote $\size{f}$ as the number of nodes in its
BDD representation.
Once a compact BDD representation is found for $f$, many problems related to the
function become tractable, with time complexity polynomial in $\size{f}$, even
if they were computationally hard originally.
Notably, counting solutions for $f(x) = 1$ using BDD requires time linear in
$\size{f}$~\cite{Knu09,Weg00}.
Therefore, when $\size{f}$ remains polynomially bounded, such counting
operations can be performed efficiently.
It is well-known that the variable ordering significantly impacts the BDD size,
and that finding the optimal ordering is known to be $\NP$-hard~\cite{Weg00}.
Most BDD packages provide dynamic variable ordering heuristics~\cite{Rud93}
which is useful for our purpose.

\section{Sum-of-Powers Representation for Quantum Circuits}\label{sec:sop-circ}


We begin our exposition by introducing the sum-of-powers (SOP) representation,
derived from the Feynman path integral formalism, which emerges as a flexible
framework for quantum circuit analysis; see,
e.g.,~\cite{Amy19a,Vil20,CBB+21,Vil22,Amy23,Vil24,DTPW24}.
The SOP approach utilizes the fact that there exist universal gate sets
consisting of unitary gates with particularly simple forms---so simple that all
non-zero entries of the unitary matrix have values proportional to the
\emph{power} $\omega^{f(x_{1}, x_{2}, \ldots, x_{k})}$, where $\omega$ is some
root of unity and $x_{1}, x_{2}, \ldots, x_{k}$ are variables labeling the input
and output wires discussed later.
This method has been explored in several quantum computing works, initially
by~\cite{DHM+05} and subsequently by~\cite{Mon17}.
Here, we review and slightly extend this approach within the SOP framework.

\subsection{Gate Set $\mathcal{Z}$}

First, consider quantum circuits over the gate set
$\mathcal{Z} = \{ \Had, \Z, \CZ, \CCZ \}$.
As in~\cite{Mon17}, we will label the wires of the circuit using the following
method.
Initially, we introduce a new variable for each qubit to label its corresponding
wire.
Subsequently, each Hadamard gate creates an additional variable to label its
output wire, which is the only way a new variable can be generated.
For a circuit with $n$ qubits and $h$ Hadamard gates, the total number of
variables will be $n + h$.
A labeling example is given in the left part in \cref{fig:labeling}, where the
circuit consists of four $\Had$ gates, one $\Z$ gate, one $\CZ$ gate, and two
$\CCZ$ gates.
In this circuit, seven variables $x_{1}, x_{2}, \ldots, x_{7}$ label the circuit
wires where $x_{1}, x_{2}, x_{3}$ represent the initial three qubits and
$x_{4}, \ldots, x_{7}$ are introduced by the four $\Had$ gates.

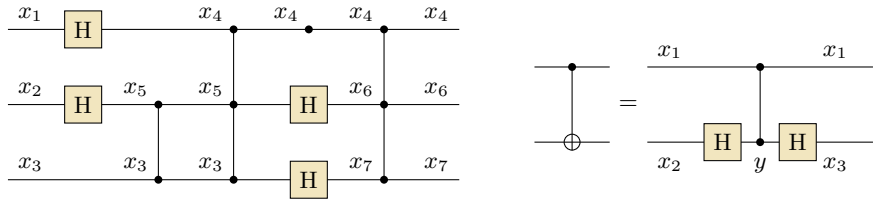
\begin{figure*}[htbp!]
  \centering
    \begin{tikzpicture}
      \node (In2) at (0, 2) [left] {};
      \node (In1) at (0, 1) [left] {};
      \node (In0) at (0, 0) [left] {};

      \node (Out2) at (6, 2) [right] {};
      \node (Out1) at (6, 1) [right] {};
      \node (Out0) at (6, 0) [right] {};

      \draw (In0) -- (Out0);
      \draw (In1) -- (Out1);
      \draw (In2) -- (Out2);

      \node (x1) at (.3, 2.2) {$x_{1}$};
      \node (x2) at (.3, 1.2) {$x_{2}$};
      \node (x3) at (.3, 0.2) {$x_{3}$};

      \node (x5) at (1.7, 1.2) {$x_{5}$};
      \node (x3) at (1.7, 0.2) {$x_{3}$};

      \node (x4) at (2.7, 2.2) {$x_{4}$};
      \node (x5) at (2.7, 1.2) {$x_{5}$};
      \node (x3) at (2.7, 0.2) {$x_{3}$};

      \node (x4) at (4.7, 2.2) {$x_{4}$};
      \node (x6) at (4.7, 1.2) {$x_{6}$};
      \node (x7) at (4.7, 0.2) {$x_{7}$};

      \node (x4) at (3.7, 2.2) {$x_{4}$};

      \node (x4) at (5.7, 2.2) {$x_{4}$};
      \node (x6) at (5.7, 1.2) {$x_{6}$};
      \node (x7) at (5.7, 0.2) {$x_{7}$};

      \node[gate] (H) at (1, 2) {$\Had$};
      \node[gate] (H) at (1, 1) {$\Had$};

      \node[control] (C0) at (2, 0) {};
      \node[control] (C1) at (2, 1) {};
      \draw (C0.center) -- (C1.center);

      \node[control] (C0) at (3, 0) {};
      \node[control] (C1) at (3, 1) {};
      \node[control] (C2) at (3, 2) {};
      \draw (C0.center) -- (C2.center);

      \node[gate] (H) at (4, 0) {$\Had$};
      \node[gate] (H) at (4, 1) {$\Had$};
      \node[control] (C0) at (4, 2) {};

      \node[control] (C0) at (5, 0) {};
      \node[control] (C1) at (5, 1) {};
      \node[control] (C2) at (5, 2) {};
      \draw (C0.center) -- (C2.center);

      \begin{scope}[xshift=7cm, yshift=.5cm]
        \node (In1) at (0, 1) [left] {};
        \node (In0) at (0, 0) [left] {};

        \node (Out1) at (1, 1) [right] {};
        \node (Out0) at (1, 0) [right] {};

        \draw (In0) -- (Out0);
        \draw (In1) -- (Out1);

        \node[control] (C) at (.5, 1) {};
        \node[target] (T) at (.5, 0) {};
        \draw (C.center) -- (T.south);

        \node (eq) at (1.25, .5) {$=$};
      \end{scope}

      \begin{scope}[xshift=8.5cm, yshift=.5cm]
        \node (In1) at (0, 1) [left] {};
        \node (In0) at (0, 0) [left] {};

        \node (Out1) at (3, 1) [right] {};
        \node (Out0) at (3, 0) [right] {};

        \draw (In0) -- (Out0);
        \draw (In1) -- (Out1);

        \node[gate] (H) at (1, 0) {$\Had$};
        \node[gate] (H) at (2, 0) {$\Had$};
        \node[control] (C0) at (1.5, 0) {};
        \node[control] (C1) at (1.5, 1) {};
        \draw (C0.center) -- (C1.center);

        \node (x1) at (.3, 1.2) {$x_{1}$};
        \node (x1) at (2.5, 1.2) {$x_{1}$};
        \node (x2) at (.3, -.3) {$x_{2}$};
        \node (x3) at (2.5, -.3) {$x_{3}$};

        \node (y) at (1.5, -0.3) {$y$};
      \end{scope}
    \end{tikzpicture}
    \caption{Left: Variable labeling for a quantum circuit of gates from gate
      set $\mathcal{Z}$.
      Right: A decomposition that derives a power-of-sum representation for
      $\CNOT$.}\label{fig:labeling}
\end{figure*}

The $\Had$ gate has matrix entries ${(-1)}^{xy} / \sqrt{2}$ where $x$ and $y$ are the
variables for the input and output wires respectively.
The $\CCZ$ gate has diagonal entries of ${(-1)}^{x_{1} x_{2} x_{3}}$ with $0$
entries elsewhere, using the same set of variables for the input and output to
indicate its diagonal nature.
In the example of \cref{fig:labeling}, the top qubit's $\Had$ gate is
represented as the power ${(-1)}^{x_{1} x_{4}} / \sqrt{2}$ and the final $\CCZ$
gate is represented by ${(-1)}^{x_{4} x_{6} x_{7}}$.
The gates $\CZ$ and $\Z$ can be discussed similarly with fewer variables.
Notice that, in this example, all gates in $\mathcal{Z}$ are represented by
powers of $-1$ up to a normalization factor (for $\Had$).

The power forms of the gates in the set $\mathcal{Z}$ can be identified as a
special tensor.
For $\Had$, it is a tensor of two legs labeled by variables $x, y$, and for
given values $x, y \in \bin$, the tensor takes value ${(-1)}^{xy}/\sqrt{2}$.
The gate $\CCZ$ is a tensor of six legs where three of them
$x_{1}, x_{2}, x_{3}$ are input and other three $y_{1}, y_{2}, y_{3}$ are
output.
This tensor takes value zero if input variables $(x_{1}, x_{2}, x_{3})$ and
output variables $(y_{1}, y_{2}, y_{3})$ differ, and
${(-1)}^{x_{1} x_{2} x_{3}}$ when they are identical.
This explains the reason we use the same set of variables for $\CCZ$ in the
labeling which enforces the input-output variable equality.
Consequently, each gate in $\mathcal{Z}$ is mapped to a tensor of the power
form, and each variable represents either the external (input and output)
variables or internal variables that are summed over by the tensor network
contraction operation.
For example, the circuit in the left of \cref{fig:labeling} corresponds to a
tensor
$\frac{1}{\sqrt{2^{4}}} \sum_{x_{5}} {(-1)}^{f_{C}(x_{1}, x_{2}, \ldots, x_{7})}$
where
\begin{equation*}
  \begin{split}
    & f_{C}(x_{1}, x_{2}, \ldots, x_{7}) = x_{1} x_{4} + x_{2} x_{5} + x_{3} x_{5} \\
    & \qquad\qquad + x_{3} x_{4} x_{5} + x_{4} + x_{5} x_{6} + x_{3} x_{7} + x_{4}
      x_{6} x_{7}.
  \end{split}
\end{equation*}
We call this summation the \emph{sum-of-powers} form for the circuit.
In this example, $x_{5}$ is the only variable summed over in the expression.

Generally, for a circuit using gates from set $\mathcal{Z}$, the sum-of-powers
form is $\frac{1}{\sqrt{2^{h}}} \sum_{y} {(-1)}^{f_{C}(x,y)}$, where $h$ is the
number of $\Had$ gates and $y$ stands for the internal variables.
We define $-1$ as the \emph{base} of the sum-of-powers forms above and
\emph{modulus} as the smallest positive integer exponent that brings the
base to identity.
The function $f_{C}$ is a multilinear polynomial of degree at most three.
For each $\Had$ gate, the product of the input and output variables gives the
term.
While for $\Z$, $\CZ$, or $\CCZ$ gates, the product of the input variables
gives the term.


\subsection{Complex Gates}

Not all elementary gates have the power form, but fortunately, there are
different ways to work with them in the sum-of-powers framework.
Consider the $\CNOT$ gate as an example, where
$\CNOT: \ket{x_{1}, x_{2}} \mapsto \ket{x_{1}, x_{1} \oplus x_{2}}$.
By definition, $\CNOT$ is a tensor that evaluates to $1$ if the input
variables $x_{1}, x_{2}$ and output variables $y_{1}, y_{2}$ satisfy
$y_{1} = x_{1}$ and $y_{2} = x_{1} \oplus x_{2}$ and to $0$ otherwise.
So one way to work with the $\CNOT$ gate is to label the wires with not just
variables, but linear functions of variables.
This approach was taken in~\cite{DHM+05}, one of the early papers that used the
wire labeling variables to analyze quantum circuits.

We take a different approach and use the fact that gates in $\mathcal{Z}$ are
universal so one can decompose $\CNOT$ as
$(\Id \otimes \Had) \, \CZ \, (\Id \otimes \Had)$.
See the right part of \cref{fig:labeling} for an illustration.
A sum-of-powers form for $\CNOT$ is, therefore,
$\frac{1}{2} \sum_{y} {(-1)}^{x_{1}y + x_{2}y + x_{3}y}$.
It is easy to verify directly that if $x_{3} = x_{1} \oplus x_{2}$, the
summation evaluates to $1$ and otherwise evaluates to $0$, which is consistent
with the definition of $\CNOT$.
So the $\CNOT$ gate in the sum-of-powers form is represented by input variables
$x_{1}$ and $x_{2}$, output variables $x_{1}$ and $x_{3}$, internal variable $y$
and normalization factor $2$.
In this approach, gates obtained by such decompositions as the $\CNOT$ gate are
called complex gates, whereas gates of the power form are called simple gates.
The complex gates can be treated easily as the ``syntactic sugar'' of a sequence
of simple gates in implementation.

We prefer this method because it is flexible and is much easier to implement.
For example, it works with $\Toffoli$ gates as well where the $\Toffoli$ gate
has input variables $x_{1}, x_{2}, x_{3}$, output variables
$x_{1}, x_{2}, x_{4}$, and one internal variable $y$.
The sum-of-powers form for $\Toffoli$ is
$\frac{1}{2} \sum_{y} {(-1)}^{x_{1}x_{2} y + x_{3} y + x_{4} y}$.
If we choose the first method using functions of variables as wire labels,
multiple $\Toffoli$ gates will induce high-degree polynomials in the general
case, which may be even more difficult to implement.

\subsection{Gate Set $\mathcal{T}$}

All the sum-of-powers forms we have considered so far have modulus $2$ and the
functions $f_{C}$ are effectively modulo $2$.
We emphasize that the modulus pertains to a sum-of-powers form of a gate, not
the gate itself, as a single gate may admit multiple sum-of-powers
representations with varying moduli.


We now expand our scope to the gate set $\mathcal{T} = \{\CNOT, \Had, \T\}$ and
introduce sum-of-powers forms of base other than $-1$.
Let $\omega_{8}$ be the $8$-th root of unity $\omega_{8} = \e^{i\pi/4}$.
The $\T$ gate is a diagonal matrix
$\begin{pmatrix} 1 & 0\\ 0 & \omega_{8} \end{pmatrix}$.
So in the sum-of-powers form, the input and output wire will have the same
variable $x$ and the $\T$ gate is represented as the power $\omega^{x}_{8}$.
The $\Had$ gate has input and output variables $x$ and $y$, respectively, as
before, and has a power form $\omega^{4xy}_{8} / \sqrt{2}$.
The $\CNOT$ gate has input $x_{1}, x_{2}$ and output variables $x_{1}, x_{3}$ as
before and take the form
$\frac{1}{2} \sum_{y} \omega^{4 (x_{1} + x_{2} + x_{3}) y}_{8}$.
In sum-of-powers representations, we work with the same base for all gates in
the gate set.
This is always possible to achieve by choosing the least common multiple of all
the moduli of the gates as the common modulus, as we have done for the gate set
$\mathcal{T}$.


\subsection{Gate Set $\mathcal{G}$}
In addition to the gate sets $\mathcal{Z}$ and $\mathcal{T}$, the sum-of-powers
framework is flexible enough to natively represent the gate set used in the
Google supremacy experiment as well~\cite{AAB+19}.
There, the gate set employed is $\mathcal{G}$ which includes the single-qubit
gates
$ \sqrt{\X} = \frac{1}{\sqrt{2}}
  \begin{pmatrix}
    1 & -i\\
    -i & 1
  \end{pmatrix},\quad
  \sqrt{\Y} = \frac{1}{\sqrt{2}}
  \begin{pmatrix}
    1 & -1\\
    1 & 1
  \end{pmatrix},\quad
  \sqrt{\W} = \frac{1}{\sqrt{2}}
  \begin{pmatrix}
    1 & -\sqrt{i}\\
    \sqrt{-i} & 1
  \end{pmatrix}$,
and two-qubit gates
\begin{equation*}
  \fSim (\pi/2, \pi/6) =
  \begin{pmatrix}
    1 & & & \\
    & 0 & -i & \\
    & -i & 0 & \\
    & & & \e^{-i\pi/6}
  \end{pmatrix},\quad
  \iSWAP =
  \begin{pmatrix}
    1 & & & \\
    & 0 & -i & \\
    & -i & 0 & \\
    & & & 1
  \end{pmatrix}.
\end{equation*}
The $\iSWAP$ is a simplified version of $\fSim(\pi/2, \pi/6)$ often used in
benchmarking BDD simulation methods.
All these gates have the power form shown in \cref{tab:google}, and the common
modulus is $24$.

\begin{table}[htbp!]
  \fontsize{6pt}{8}\selectfont
  \centering
  \caption{Sum-of-powers representation for the Google supremacy gate
    set.}\label{tab:google}
  \begin{tabular}{lllll}
    \toprule
    Gate & Input & Output & Representation & Factor \\
    \midrule
    $\sqrt{\X}$ & $x_{0}$ & $x_{1}$
      & $\omega^{18 x_{0} + 18 x_{1} + 12 x_{0} x_{1}}_{24}$ & $1 / \sqrt{2}$\\
    $\sqrt{\Y}$ & $x_{0}$ & $x_{1}$
      & $\omega^{12 x_{0} + 12 x_{0} x_{1}}_{24}$ & $1 / \sqrt{2}$\\
    $\sqrt{\W}$ & $x_{0}$ & $x_{1}$
      & $\omega^{15 x_{0} + 21 x_{1} + 12 x_{0} x_{1}}_{24}$ & $1 / \sqrt{2}$\\
    $\fSim(\pi/2, \pi/6)\quad$ & $x_{0}, x_{1}\quad$ & $x_{1}, x_{0}\quad$
      & $\omega^{18 x_{0} + 18 x_{1} + 10 x_{0} x_{1}}_{24}\quad$ & $1$\\
    $\iSWAP$ & $x_{0}, x_{1}\quad$ & $x_{1}, x_{0}\quad$
      & $\omega^{18 x_{0} + 18 x_{1} + 12 x_{0} x_{1}}_{24}\quad$ & $1$\\
    \bottomrule
  \end{tabular}
\end{table}

We have introduced three distinct gate sets---$\mathcal{Z}$, $\mathcal{T}$, and
$\mathcal{G}$---each characterized by its native gates.
We summarize the discussion in the following theorem.
\begin{theorem}\label{thm:sop}
  For any quantum circuit $C$ of $n$ qubits and $m$ gates in universal gate sets
  $\mathcal{Z}$, $\mathcal{T}$, or $\mathcal{G}$, one can efficiently derive an
  SOP form $\frac{1}{\sqrt{R}} \sum_{y} \omega^{f(x,y)}$, representing the
  tensor of the circuit.
  In the representation, $f(x,y)$ is a multilinear polynomial of $\order{m}$
  terms and degree at most three, and $x$ corresponds to external variables
  assigned to the input and output wires of $C$.
\end{theorem}

A key advantage of our method is its remarkable flexibility: the approach can be
readily extended to new gate sets without requiring customized implementation,
in contrast to previous methods like~\cite{TJJ21}.
Supporting a new gate set becomes a simple matter of creating a configuration
file that defines the coefficients, monomial terms, and other characteristics of
the sum-of-powers representation for each gate in the set.

\section{Tensor Contraction and Substitution on Sum-of-Powers}\label{sec:sop-tensor}

In the circuit simulation and equivalence checking problems discussed below, it is
convenient to work with the sum-of-powers representations not only for circuits,
but for other derived quantities such as $\bra{a} U_{C} \ket{0^{n}}$ as well.
These quantities are easy to formulate in the sum-of-powers framework given the
following two operations of SOPs.

In tensor networks, the most important procedure for manipulating tensors is the
tensor contraction operation, where two indices are identified and summed
over~\cite{MS08}.
For a sum-of-powers tensor $\sum_{y} \omega^{f(x, y)}$ and two external
variables $x_{1}$ and $x'_{1}$ in $x$, it is natural to consider the
\emph{contraction} of variables $x_{1}$ and $x'_{1}$.
The resulting form is still a sum-of-powers tensor having the form
$\sum_{x_{1}, y} \omega^{f[x_{1}/x'_{1}](x, y)}$, where $f[x_{1}/x'_{1}]$ is a
function obtained by substituting all the variable $x'_{1}$ in $f$ with $x_{1}$.
More generally, let $x_{1}, x_{2}, \ldots, x_{k}$ and
$x'_{1}, x'_{2}, \ldots, x'_{k}$ be the different external variables that are
contracted respectively, then the resulting sum-of-powers tensor is
\begin{equation*}
  \sum_{x_{1}, \ldots, x_{k}, y} \omega^{f[x_{1}/x'_{1},\, \ldots,\, x_{k}/x'_{k}](x, y)},
\end{equation*}
where $f[x_{1}/x'_{1}, \ldots, x_{k}/x'_{k}]$ is the function obtained by
substituting $x'_{1}, x'_{2}, \ldots, x'_{k}$ with $x_{1}, x_{2}, \ldots, x_{k}$
respectively.
Note that even though we write
$f[x_{1}/x'_{1}, \ldots,\allowbreak x_{k}/x'_{k}](x, y)$ but
$f[x_{1}/x'_{1}, \ldots, x_{k}/x'_{k}]$ is a function independent of
$x_{1}', x'_{2}, \ldots, x'_{k}$.

Another procedure we need is variable \emph{substitution}.
Let $\sum_{y} \omega^{f(x, y)}$ be a sum-of-powers tensor, and let
$x_{1}, x_{2}, \ldots, x_{k}$ be external variables in $x$.
For any values $a_{1}, a_{2}, \ldots, a_{k} \in \bin$, the substitution of
$x_{i}$ with $a_{i}$ for $i=1, 2, \ldots, k$ results in a sum-of-powers tensor
$\sum_{y} \omega^{f[a_{1}/x_{1}, \ldots, a_{k}/x_{k}](x, y)}$.
This operation naturally arises when the wire corresponding to variable
$x_{j}$ is contracted with a basis state $\ket{a_{j}}$.
The above discussions establish the following simple claim.
\begin{theorem}\label{thm:sop-op}
  The contraction and substitution of an SOP result in another SOP\@.
\end{theorem}

With the contraction and substitution operations established, we can now
efficiently compute the sum-of-powers forms for the two quantities we require
below.
The first is the amplitude $\bra{a} U_{C} \ket{0^{n}}$ where $U_{C}$ is the
unitary matrix for circuit $C$ and $a \in \bin^{n}$ is a computation basis.
Let $\frac{1}{R} \sum_{y} \omega^{f_{C}(x, y)}$ be the sum-of-powers tensor for
circuit $C$.
From \cref{thm:sop}, we know that $f_{C}$ is a summation of $\order{m}$ monomial
terms where $m$ is the number of gates in the circuit.

Let $x^{\lin} = (x^{\lin}_{1}, x^{\lin}_{2}, \ldots, x^{\lin}_{n})$ and
$x^{\lout} = (x^{\lout}_{1}, x^{\lout}_{2}, \ldots, x^{\lout}_{n})$ be the
variables corresponding to the input and output qubits respectively.
Each consists of $n$ different variables, but they may share some common
variables.
This could happen, for example, when all gates acting on a qubit are diagonal,
and the labeling strategy discussed in \cref{sec:sop-circ} will not introduce
new variables for the qubit.
Suppose there are $s \ge 0$ common variables in $x^{\lin}$ and $x^{\lout}$ and
there are qubit indices $j_{1}, \ldots, j_{s}$ and $k_{1}, \ldots, k_{s}$ such
that $x^{\lin}_{j_{i}} = x^{\lout}_{k_{i}}$ for all $i=1, \ldots, s$.
Let $k_{s+1}, \ldots, k_{n}$ be the output qubit indices assigned non-common
variables (that is, they are not in $\{k_{1}, \ldots, k_{s}\}$).
The sum-of-powers tensor for $\bra{a} U_{C} \ket{0^{n}}$ is then $0$ if there is
an $i \in \{1, 2, \ldots, s\}$ such that $a_{k_{i}} \ne 0$ (in other words, this
constitutes a contradictory substitution, as a variable cannot simultaneously be
substituted with both $0$ and $1$), or $\frac{1}{R} \sum_{y} \omega^{f'_{C}(y)}$
where
\begin{equation*}
  f'_{C}(y) = f_{C}[a_{k_{s+1}} / x^{\lout}_{k_{s+1}}, \ldots,
  a_{k_{n}} / x^{\lout}_{k_{n}}, 0 / x^{\lin}_{1}, \ldots, 0 / x^{\lin}_{n}](y).
\end{equation*}
This finishes the discussion on how to represent $\bra{a} U_{C} \ket{0^{n}}$ as
a sum-of-powers tensor network.

Next, we consider the representation for $\tr U_{C}$.
The basic idea is quite simple; we only need to perform a contraction of the
input and output variables.
Since an input variable may become an output variable on another qubit (e.g.,
the $\iSWAP$ gate in the gate set $\mathcal{G}$), we must ensure consistent
substitution of each original variable during the contraction.
Continuing with the above setup, consider the case that there are $s$ common
varibles.
Define a graph $G$ of $n$ vertices containing $s$ edges (or self-loops)
specified by $\{j_{i}, k_{i}\}$.
Consider partitioning of the graph into $t$ connected components
$G_{1}, G_{2}, \ldots, G_{t}$.
Define a new variable $z_{j}$ for each $G_{j}$ where $j=1, 2, \ldots, t$.
Suppose circuit $C$ has the sum-of-powers form
$\frac{1}{\sqrt{R}} \sum_{y} \omega^{f(x, y)}$.
It is easy to convince oneself that the sum-of-powers form for $\tr U_{C}$ is
$\frac{1}{\sqrt{R}} \sum_{z, y} \omega^{f'_{C}(z, y)}$ where $f'$ is the
function obtained by substituting all external variables $x$ in $f_{C}$ using
the following method.
For each input variable $x^{\lin}_{i}$ and output variable $x^{\lout}_{i}$, let
$G_{j_{i}}$ be the connected component to which vertex $i$ belongs, and we
substitute $x^{\lin}_{i}$ and $x^{\lout}_{i}$ with $z_{j_{i}}$.

\section{FeynmanDD:\ Decision Diagram for Sum-of-Powers}\label{sec:feynmandd}


In \cref{sec:sop-circ,sec:sop-tensor}, we presented methods for deriving
sum-of-powers representations for circuit and quantities related to the circuit
$\frac{1}{\sqrt{R}} \sum_{y} \omega^{f(x, y)}$, where $R$ is a normalization
factor, $\omega$ is the $r$-th root of unity, and
$f : \bin^{k} \rightarrow \{0, 1, \ldots, r-1\}$ is a multilinear polynomial
with values modulo $r$.
The next step of the FeynmanDD method is to represent the function $f$ using a
variant of BDD called multi-terminal binary decision diagram (MTBDD).


As discussed in \cref{sec:intro}, FeynmanDD departs from most existing
BDD-motivated approaches by utilizing classical decision diagrams as its
underlying data structure.
This strategy facilitates the immediate use of robust and efficient
implementations developed over decades, including CUDD~\cite{Som05},
BuDDy~\cite{Lin99}, Sylvan~\cite{Dij16}, and adiar~\cite{SPJT21}.
Moreover, it seamlessly incorporates powerful variable ordering heuristics.
CUDD and Sylvan offer multi-terminal support, rendering them suitable MTBDD
engines for gate sets of any modulus.
BuDDy and adiar currently support only standard BDDs and are consequently
limited to gate sets with a modulus of $2$.
Our current implementation exclusively employs CUDD, leaving support for
additional packages reserved for future development.


In many applications, it is possible to first represent the function $f$ for the
circuit $C$ as an MTBDD and then try to compute the MTBDD for function $f'$ for
the derived quantities such as $\bra{a} U_{C} \ket{0^{n}}$ or $\tr U_{C}$.
But the substitution and contraction usually help to simplify the multilinear
polynomial significantly in the sum-of-powers and thereby also simplify its
MTBDD representation.
It is therefore preferred first to try to use methods in \cref{thm:sop-op}
to transform $f$ at the polynomial representation level and delay the creation
of the MTBDD as late as possible.

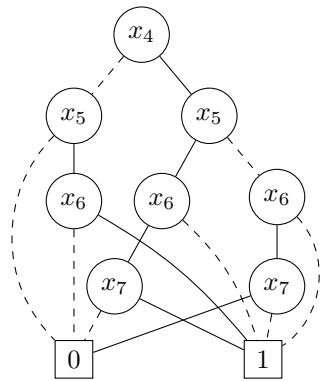
\begin{wrapfigure}[18]{r}{.4\textwidth}
  \centering
  \begin{tikzpicture}[node distance=2cm,scale=0.9,
    state/.style = {draw, circle},
    dots/.style = {dashed},
    terminal/.style = {draw, minimum width=.5cm, minimum height=.5cm}]

    \node[terminal] (zero) at (0,0) {$0$};
    \node[terminal, right=of zero] (one) {$1$};
    \node[state] (71) at (.6, 1.08) {$x_{7}$};
    \node[state] (72) at (3, 1.08) {$x_{7}$};
    \node[state] (61) at (0, 2.34) {$x_{6}$};
    \node[state] (62) at (1.3, 2.34) {$x_{6}$};
    \node[state] (63) at (3, 2.4) {$x_{6}$};
    \node[state] (51) at (0, 3.6) {$x_{5}$};
    \node[state] (52) at (2, 3.6) {$x_{5}$};
    \node[state] (4) at (1, 4.8) {$x_{4}$};

    \draw (4) -- (52);
    \draw (51) -- (61);
    \draw (52) -- (62);
    \draw[bend left=10] (61) to (one);
    \draw (62) -- (71);
    \draw (63) -- (72);
    \draw (71) -- (one);
    \draw (72) -- (zero);

    \draw[dots] (4) -- (51);
    \draw[dots, bend right=45] (51) to (zero);
    \draw[dots] (52) -- (63);
    \draw[dots] (61) to (zero);
    \draw[dots, bend left=15] (62) to (one);
    \draw[dots, bend left=50] (63) to (one);
    \draw[dots] (71) -- (zero);
    \draw[dots] (72) -- (one);
  \end{tikzpicture}
  \caption{BDD for the reduced function
    $x_{4} + x_{5} x_{6} + x_{4} x_{6} x_{7}$.}\label{fig:bdd}
\end{wrapfigure}

For instance, in the example circuit on the left part of \cref{fig:labeling},
the function $f_{C}$ has $7$ variables and $8$ terms.
When representing it as a BDD, it has $28$ nodes, including the constants $0$ and
$1$.
However, when enforcing the initial state condition, variables
$x_{1}, x_{2}, x_{3}$ are set to $0$, and the function reduces to one with
$4$ variables and $3$ terms: $x_{4} + x_{5} x_{6} + x_{4} x_{6} x_{7}$.
Its BDD is much smaller, containing only $10$ nodes as shown in \cref{fig:bdd}.


Given a function $f:\bin^{n} \rightarrow [r]$ representing a multilinear
polynomial of $m$ terms
\begin{equation*}
  f(x) = \sum_{k=1}^{m} a_{k}\, x_{1}^{i_{k,1}} x_{2}^{i_{k,2}} \cdots x_{n}^{i_{k,n}},
\end{equation*}
for $i_{k,j} \in \bin$, we need to build the BDD representation of it from
scratch.
This represents the most computationally expensive step in our method.
While the summation is commutative, the order of term addition can lead to
dramatically different computational costs.
To illustrate this phenomenon, consider a sequence of Hadamard gates.
The polynomial $f$ takes the form $\frac{r}{2} \sum_{j=0}^{k-1} x_{j} x_{j+1}$.
When variables are ordered sequentially, the final BDD has a size of
$\order{k}$.
However, naive sequential term addition can result in a quadratic
$\Omega(k^{2})$ runtime, becoming inefficient for large $k$.
The quadratic scaling emerges from the computational complexity of adding each
term.
Specifically, the time cost for combining the partial sum with the current term
scales linearly with the MTBDD size of the partial sum.
This linear scaling leads to a quadratic overall runtime through recursive
addition.
To address this efficiency challenge, we employ a \emph{binary synthesis
  method}.
This approach, utilizing the commutativity of addition modulo $r$, partitions
terms into two groups of approximately equal size, first computes results within
each group and then combines them.
For the Hadamard sequence example, this reduces time complexity to
$\order{k \log k}$.
Numerical simulations suggest this binary synthesis method performs effectively
across more generalized scenarios.


For the construction of a MTBDD, the second order that deserves careful
consideration is the order of the variables in the decision diagram data
structure as it may significantly impact the final size.
Consider again the Hadamard sequence example: ordering the odd variables
$x_{1}, x_{3}, \ldots$ first will cause the BDD size to grow exponentially.
A well-know dynamic variable reordering heuristic is called
sifting~\cite{Rud93}, which is supported by CUDD and is sometimes useful in the
circuit classes used in \cref{sec:exp}.
In our numerical experiments, we explore several other variable ordering
heuristics to mitigate this challenge.
For low-depth circuits, a pragmatic approach involves sequential qubit-based
ordering.
This method prioritizes variables by qubit, ordering all variables for the first
qubit, then the second, and so on.
We term this the \emph{qubit order} strategy.
For circuits with a small number of qubits, the \emph{gate order} heuristics
offer an effective variable ordering strategy.
This approach ranks variables based on the first gate that employs them,
providing a temporally informed sequence.
An alternative method leverages standard tensor contraction optimization tools
like \emph{cotengra} to determine an optimal term order.
Once this order is established, variables are arranged according to their
appearance in the ordered terms.
This approach, which we call the \emph{tensor order} can sometimes be
advantageous for random circuits.
The concept of using tensor contraction complexity to bound BDD size has
precedent in classical verification, where researchers previously explored the
treewidth of CNF formulas~\cite{FPV05}.
Our approach diverges by focusing on the treewidth of XOR formulas.
By employing these sophisticated ordering strategies, we can significantly
optimize the efficiency of the MTBDD creation step.



The intuition behind the tensor contraction based heuristics is as follows.
A fundamental fact about the size of a Boolean function $f$ given variable order
$x_{1}, x_{2}, \ldots, x_{n}$ is that the number of nodes with label $x_{i}$ is
the different number of functions $f[a_{1}/x_{1}, \ldots, a_{i-1}/x_{i-1}]$ that
essentially depends on $x_{i}$.
This number is upper bounded by $O\bigl( 2^{N_{i}} \bigr)$ where $N_{i}$ is the
number of external wires in the process of contracting the sum-of-powers tensor
network.
While it is a useful ordering heuristics, its performance on linear-network
circuits is poor, indicating that the tensor contraction complexity bound is an
upper bound that may be very loose in certain cases.


Another useful technique for improving the efficiency of FeynmanDD is to use the
counting identity $\sum_{x} {(-1)}^{xy+xz} = 2\, \delta_{y,z}$ to simplify the
function $f$ in the SOP\@.
More specifically, when a variable $x$ appears exactly twice in the terms having
form $\frac{r}{2} xx_{0}$ and $\frac{r}{2} xx_{1}$, we have
\begin{equation*}
  \frac{1}{\sqrt{R}} \sum_{x, x_{0}, x_{1}, y} \omega^{f(x,x_{0},x_{1},y)}
  = \frac{1}{\sqrt{R}} \sum_{x_{0}, x_{1}, y} \sum_{x}
      {(-1)}^{x(x_{0} + x_{1})} \omega^{f_{\text{rem}}(x_{0}, x_{1}, y)},
\end{equation*}
which simplifies to
$\frac{2}{\sqrt{R}} \sum_{x_{0}, y} \omega^{f_{\text{rem}}[x_{0}/x_{1}](x_{0},y)}$.
This technique reduces the number of variables and therefore the size of the
MTBDD and is most useful for simulating simple structured circuits.


On the technical level, CUDD does not support addition modulo $r$ by default,
and we implement a custom-made function for this and use CUDD's API
$\texttt{Cudd\_addApply}$ to perform the operation on the data structure it
maintains.
Another technical issue worth mentioning is the requirement to handle high
precision integers, as the counters involved for large circuits is huge.
We implement a high precision counting algorithm for CUDD using the \emph{GNU
  multiple precision arithmetic library}.

\section{Quantum Circuit Analysis using FeynmanDD}\label{sec:app}

\subsection{Simulation of Quantum Circuits}\label{sec:simulation}


In quantum circuit simulation, a fundamental problem is computing the amplitude
$\bra{a} C \ket{0}$ for a given $a \in \bin^{n}$.
This single amplitude simulation problem, which appears straightforward, is
actually $\BQP$-complete even for an approximation and $a = 0^{n}$.
Considering a multilinear polynomial $f$ representing the sum-of-powers
representation of $\bra{a} C \ket{0}$ with size $\size{f}$, our simulation
algorithm achieves runtime linear in $\size{f}$.
Consequently, when the multi-terminal binary decision diagram (MTBDD)
representing $f$ is small, we can effectively solve this problem.

The core algorithm reduces the evaluation of the sum-of-powers form to a finite
number of counting problems on the MTBDD\@.
By rewriting the amplitude SOP representation, we can express it as
\begin{equation*}
  \frac{1}{\sqrt{R}} \sum_{y} \omega^{f(y)} =
  \frac{1}{\sqrt{R}} \sum_{j=0}^{r-1} N_{j} \omega^{j},
\end{equation*}
where $N_{j} = \bigl| \bigl\{ y \mid f(y) \equiv j \pmod{r} \bigr\} \bigr|$
represents the number of inputs to $f$ that evaluate to
$j \in \{0, 1, \ldots, r-1\}$.
Given the MTBDD representing $f$, there are algorithms that can count the
numbers $N_{j}$ in time $\order{m \size{f}}$ where $m$ is the bit length of the
counting result~\cite{Knu09,Weg00}.
The number $m$ is bounded by the number of gates, but usually much smaller and
can often be considered constant in practice.


In quantum circuit simulation, another common task is computing the acceptance
probability for a quantum circuit.
Consider a circuit $C$ acting on the initial state $\ket{0^{n}}$, with its
sum-of-powers representation expressed as $\frac{1}{\sqrt{R}}\sum_{y} \omega^{f(x,y)}$,
where $x$ represents variables labeling the output qubits, the task is to
estimate the probability of observing $1$ when measuring the first qubit.
We introduce a derived function $F(x, y, y') = f(x, y) - f(x, y')$ that operates
on variables $x, y, y'$, with $y'$ being a new set of variables matching $y$'s
size.
By defining $x_{1}$ as the variable of the qubit to be measured and $x_{>1}$ as
the remaining variables, the acceptance probability can be formulated as:
\begin{equation*}
  \begin{split}
    \Pr(C \text{ accepts})
    & = \bra{0^{n}} C^{\dagger} \bigl( \ket{1}\bra{1} \otimes I \bigr)
      C \ket{0^{n}}\\
    & = \frac{1}{R} \sum_{x_{>1}, y, y'}\omega^{f[1/x_{1}](x_{>1},y)-f[1/x_{1}](x_{>1}, y')}\\
    & = \frac{1}{R} \sum_{x_{>1}, y, y'}\omega^{F[1/x_{1}](x_{>1}, y, y')},
  \end{split}
\end{equation*}
where the contraction and substitution operations of SOPs are employed.
Hence, the probability can be represented as
$\frac{1}{R} \sum_{z} \omega^{F[1/x_{1}](z)}$ for $z = (x_{>1}, y, y')$.
Notice that $\size{F[1/x_{1}]} \le \size{F} \le \order{\size{f}^{2}}$, the
algorithm has a quadratic complexity in terms of $\size{f}$.
In practice, we indeed observe that it is much less efficient compared with the
linear complexity of the amplitude computation task.
Yet, its performance is still comparably well for many circuit families.


Using a similar technique, we can estimate joint probabilities of measuring
output qubits corresponding to variables $x_{1}, \ldots, x_{j}$ as
\begin{equation*}
  \Pr[x_{1} = a_{1}, \ldots, x_{j} = a_{j}] =
  \frac{1}{R} \sum_{x_{>j}, y, y'} \omega^{F[a_{1}/x_{1}, \ldots,
    a_{j}/x_{j}](x_{>j}, y, y')},
\end{equation*}
where $x_{>j}$ are the remaining variables in $x$ excluding
$x_{1}, \ldots, x_{j}$.
This allows us to compute the conditional probabilities using the conditional
probability formula
\begin{equation*}
  \Pr[ x_{j} = a_{j} \,|\, x_{1} = a_{1}, \ldots, x_{j-1} = a_{j-1}] =
  \frac{\Pr[x_{1} = a_{1}, \ldots, x_{j} = a_{j}]}{\Pr[x_{1} = a_{1},
    \ldots, x_{j-1} = a_{j-1}]}.
\end{equation*}
for all $j$ and $a_{1}, a_{2}, \ldots, a_{j} \in \bin$.
We can leverage this method to sample from the output distribution of
$C\ket{0^{n}}$ using a simple sequential sampling algorithm.
The process begins by computing the probability $p_{1} = \Pr[x_{1} = 0]$ and
sampling the first bit with probabilities $p_{1}$ and $1 - p_{1}$.
For subsequent bits, we compute
$p_{j} = \Pr[x_{j} = 0 \,|\, x_{1} = a_{1}, \ldots, x_{j-1} = a_{j-1}]$ and
sample accordingly.
By repeating this process for $j=2, 3, \ldots, n$, we complete the sampling of
all $n$ output bits.


We remark that even though our simulation method requires a discrete universal
gate set and cannot deal with arbitrary single-qubit rotations directly, it is
possible to work with such circuits by expanding the rotations using
Solovay-Kitaev theorem~\cite{Daw19} or methods in~\cite{GS13,RS16} to a sequence
of, for example, $\Had$ and $\T$ gates.
As $\Had$ will create a new variable and $\T$ gates will not, the gate sequence
will introduce terms of the form $\frac{r}{2} \sum_{j=1}x_{j}x_{j+1} + \ell(x)$
where $x_{2}, x_{3}, \ldots$ are new variables introduced by the $\Had$ gates
and $\ell(x)$ is a linear function.
By ordering $x_{2}, x_{3}, \ldots$ after $x_{1}$, the resulting BDD will have a
size not very sensitive to the length of the gate sequence.

\subsection{Circuit Equivalence Checking}\label{sec:equivalence}

Circuit equivalence checking is another important task for quantum circuit
analysis and optimization~\cite{PBW22,HYF+21,HFLY22}.
Usually, BDD-based methods are ideal for such a task thanks to the uniqueness of
BDD representations, and the circuit equivalence checking problem is reduced to
check whether the BDDs for two circuits are identical.
In our case, however, the BDD's for equivalent circuits may be dramatically
different as it is based on the classical syntactical description of the
circuit, not the semantic meaning (the unitary operator) of the circuit.

Fortunately, however, we are still able to use FeynmanDD to compute the trace of
the unitary operator $\tr U_{C}$ for circuit $C$.
Given two circuits $C_{0}, C_{1}$ of $n$ qubits,
$\tr(U^{\dagger}_{C_{0}} U_{C_{1}}) = 2^{n} \omega^{j}$ for some $j$, if and
only if $C_{0}$ and $C_{1}$ are equivalent.
This is a fact that was utilized in many previous works on circuit equivalence
checking~\cite{HYF+21,PBW22}.
Our method is then to first obtain the sum-of-powers form for
$\frac{1}{2^{n}} \tr(U^{\dagger}_{C_{0}} U_{C_{1}})$ and build the corresponding
FeynmanDD for it as explained in \cref{sec:feynmandd}.
And the counting based method as in \cref{sec:simulation} is used for circuit
simulation to evaluate the sum-of-powers form.
We know that the two circuits are equivalent (up to a global phase) if the value
has a unit norm.
By using the multiple precision arithmetic library, the computation is performed
exactly which is crucial for the equivalence checking application.
Numerical experiments on using FeynmanDD for equivalence checking are discussed
in \cref{sec:equivalence_check_results}.

\section{Circuit Simulation Experiments and Comparisons}\label{sec:exp}

To evaluate the performance of FeynmanDD, we conducted a series of experiments
comparing it with (a) three state-of-the-art quantum circuit simulators,
DDSIM~\cite{ZW19,ZHW19}, SliQSim~\cite{TJJ21}, and WCFLOBDD~\cite{SCR23,SCR24}
and (b) MQT-QCEC (https://github.com/cda-tum/qcec) for the task of equivalence
checking, a widely used tool that integrates multiple equivalence checking
techniques.
The experiments were performed on a server equipped with two Intel (R) Xeon (R)
Platinum 8358P CPUs @ 2.60 GHz (64 cores, 128 threads in total) and 512 GB of
memory.
However, as CUDD does not support parallel computing, each simulation is
performed using a single thread and our method barely uses more than 1 GB of
memory.

Our numerical experiments cover three types of computational tasks: (1) the
calculation of the amplitude $\bra{0} U \ket{0}$, (2) the simulation of sampling
a full computational-basis measurement outcome, and (3) the equivalence checking
of two given circuits.
Furthermore, in the amplitude and sampling tasks, we measured the runtime and
peak memory usage during the execution of these tasks by the program.
Each simulation was limited to 3600 seconds (1 hour).
For simplicity, we use qubit order to arrange the variables in all our
experiments.

We tested four different families of circuits for quantum circuit simulation
tasks: Google supremacy circuits, GHZ circuits, BV circuits, and a specially
constructed family called linear-network circuits.
The results are summarized in
\cref{tab:googleresults-amp,tab:googleresults-samp,tab:BVresults,tab:GHZresults,tab:linearnetworkresults}.
We present the comparison of circuit equivalence checking experiments in
\cref{sec:equivalence_check_results}.


\subsection{Google Supremacy Circuits}\label{sec:google_supremacy_exp}

We used benchmarks from the GRCS
repository\footnote{https://github.com/sboixo/GRCS}: \texttt{cz\_v2} and
\texttt{is\_v1} circuits.
The gate sets are
$\{ \text{H}, \sqrt{\text{X}}, \sqrt{\text{Y}}, \text{T}, \text{CZ} \}$ and
$\{ \text{H}, \sqrt{\text{X}}, \sqrt{\text{Y}}, \text{T}, \text{iSWAP} \}$
respectively.
Since SliQSim does not natively support $\text{iSWAP}$, the corresponding
results in \cref{tab:googleresults-amp,tab:googleresults-samp} were obtained via
the decomposition
\begin{equation*}
  \text{iSWAP} = (\text{S} \otimes \text{S}) (\text{H} \otimes \text{I})
  \text{CNOT}_{0, 1} \text{CNOT}_{1, 0} (\text{I} \otimes \text{H}).
\end{equation*}
For WCFLOBDD, since it does not natively support $\text{Rx}(\frac{\pi}{2})$ and
$\text{Ry}(\frac{\pi}{2})$, we decomposed them as
$\text{Rx}(\frac{\pi}{2})=\text{HSH}$ and $\text{Ry}(\frac{\pi}{2})=\text{HZ}$.
Moreover, since WCFLOBDD provides only an interface for computing probabilities
rather than amplitudes, and this interface did not function correctly in our
tests, we conducted only the \texttt{Single Sample Output} test for WCFLOBDD\@.
Each benchmark consists of ten circuits, and the table reports the average time
and memory usage across these ten circuits.

FeynmanDD significantly outperformed DDSIM, SliQSim, and WCFLOBDD in runtime and
memory usage.
For \texttt{cz\_v2} circuits, FeynmanDD was much faster and consumed less
memory.
Larger \texttt{cz/5x5\_10} circuits timed out for DDSIM, SliQSim, and WCFLOBDD,
while FeynmanDD completed within 0.05 seconds for amplitude estimation and 95
seconds for sampling.
Due to excessive memory consumption, WCFLOBDD was severely limited in this test.
Similar results were observed for \texttt{is\_v1} circuits.

\begin{table}[htbp!]
  \centering
  \fontsize{6pt}{7}\selectfont
  \caption{Quantum circuit simulation benchmarks on Google supremacy circuits
    for the task of zero to zero amplitude computation.
    In this and the following tables, $n$ stands for the number of qubits, $m$
    is the number of gates, time is in seconds (s), memory is in MB, TO means
    timeout.
    }\label{tab:googleresults-amp}
  \begin{tabular}{cccrrrrrr}
    \toprule
    \multirow{2}*{Circuit} & \multirow{2}*{$n$} & \multirow{2}*{$m$}
    & \multicolumn{6}{c}{Zero to Zero Amplitude} \\
    \cmidrule(lr){4-9} 
    & & & \multicolumn{2}{c}{DDSIM} & \multicolumn{2}{c}{SliQSim} & \multicolumn{2}{c}{Ours} \\
    \midrule
    & & & Time & Mem & Time & Mem & Time & Mem
    \\[2pt]
    cz/4x4\_10 & 16 & 115 & 0.64 & 48.0 & 5.8 & 81.3 & 0.01 & 12.1 \\
    cz/4x5\_10 & 20 & 145 & 67.6 & 362.3 & 666.9 & 255.5 & 0.03 & 12.2 \\
    cz/5x5\_10 & 25 & 184 & TO & & TO & & 0.04 & 12.2 \\
    is/4x4\_10 & 16 & 115 & 1.0 & 55.7 & {\cellcolor[gray]{0.9}10.6} & {\cellcolor[gray]{0.9}128.1} & 0.03 & 12.3 \\
    is/4x5\_10 & 20 & 145 & 158.2 & 477.5 & {\cellcolor[gray]{0.9}976.9} & {\cellcolor[gray]{0.9}245.4} & 0.1 & 15.1 \\
    is/5x5\_10 & 25 & 184 & TO & & {\cellcolor[gray]{0.9}TO} & {\cellcolor[gray]{0.9}} & 0.2 & 20.9 \\
    \bottomrule
  \end{tabular}
  \vspace{2.5em}
  \centering
  \fontsize{6pt}{7}\selectfont
  \caption{Quantum circuit simulation benchmarks on Google supremacy circuits
    for the task of sampling.
    The {\color{red} \tiny ($x$)} mark indicates that {\color{red} $x$} tests
    timeout among the ten circuits tested in the group, and the shown value is
    derived from the average of the remaining $10 - {\color{red} x}$ test
    results.}\label{tab:googleresults-samp}
  \begin{tabular}{cccrrrrrrrr}
    \toprule
    \multirow{2}*{Circuit} & \multirow{2}*{$n$} & \multirow{2}*{$m$}
    & \multicolumn{8}{c}{Single Output Sample} \\
    \cmidrule(lr){4-11} 
    & & 
     & \multicolumn{2}{c}{DDSIM} & \multicolumn{2}{c}{SliQSim} & \multicolumn{2}{c}{WCFLOBDD}& \multicolumn{2}{c}{Ours}\\
    \midrule
    & & & Time & Mem & Time & Mem & Time & Mem & Time & Mem
    \\[2pt]
    cz/4x4\_10 & 16 & 115 
    & 0.53 & 48.00 & 6.10 & 153.53 & {\cellcolor[gray]{0.9} 184.32} & {\cellcolor[gray]{0.9}18160.96} & 4.94 & 42.76\\
    cz/4x5\_10 & 20 & 145 
    & 71.65 & 362.31 & 785.91 & 355.38 & {\cellcolor[gray]{0.9}TO} & {\cellcolor[gray]{0.9}}& 34.66 & 101.94 \\
    cz/5x5\_10 & 25 & 184 
    & TO & & TO & & {\cellcolor[gray]{0.9}TO} & {\cellcolor[gray]{0.9}} & 94.63 & 132.98 \\
    is/4x4\_10 & 16 & 115 
    & 1.05 & 55.72 & {\cellcolor[gray]{0.9}12.74} & {\cellcolor[gray]{0.9}163.16} & {\cellcolor[gray]{0.9}259.91} & {\cellcolor[gray]{0.9}22349.67} & 7.62 & 76.02\\
    is/4x5\_10 & 20 & 145 
    & 178.52 & 477.57 & {\cellcolor[gray]{0.9}1020.1} & {\cellcolor[gray]{0.9}378.28} & {\cellcolor[gray]{0.9}TO} & {\cellcolor[gray]{0.9}} & 25.36 & 116.72 \\
    is/5x5\_10 & 25 & 184 
    & TO & & {\cellcolor[gray]{0.9}TO} & {\cellcolor[gray]{0.9}} & {\cellcolor[gray]{0.9}TO} &{\cellcolor[gray]{0.9}} & 153.20{\color{red} \tiny (1)} & 156.15\\
    \bottomrule
  \end{tabular}
\end{table}

\subsection{GHZ and BV Circuits}\label{sec:ghz_bv_exp}

We tested two types of circuits in this part: GHZ circuits and BV circuits.
GHZ circuits generate $n$-qubit GHZ states, while BV circuits implement the
Bernstein-Vazirani algorithm with the all-ones secret string.
These circuits were chosen as they generate simple yet highly entangled states.

\begin{table}[htbp!]
  \centering
  \fontsize{6pt}{7}\selectfont
  \caption{Quantum circuit simulation benchmarks on BV circuits.
    Time is measured in seconds (s), memory usage is meausred in
    MB.}\label{tab:BVresults}
  \begin{tabular}{ccrrrrrrrrrrrrrr}
    \toprule
    \multirow{2}*{Qubits} & \multirow{2}*{Gates}
    & \multicolumn{6}{c}{Zero to Zero Amplitude}
    & \multicolumn{8}{c}{Single Output Sample} \\
    \cmidrule(lr){3-8} \cmidrule(r){9-16}
    & & \multicolumn{2}{c}{DDSIM} & \multicolumn{2}{c}{SliQSim} & \multicolumn{2}{c}{Ours}
    & \multicolumn{2}{c}{DDSIM} & \multicolumn{2}{c}{SliQSim} & \multicolumn{2}{c}{WCFLOBDD}& \multicolumn{2}{c}{Ours}\\
    \midrule
    & & Time & Mem & Time & Mem & Time & Mem & Time & Mem & Time & Mem & Time & Mem & Time & Mem\\[2pt]
    100 & 299 & 0.002 & 71.4 & 0.05 & 14.5 & 0.005 & 12.2
    & \multicolumn{2}{c}{\cellcolor[gray]{0.9}} & 0.05 & 14.62 & 0.04 & 480.02 & 0.007 & 14.54 \\
    500 & 1499 & 0.04 & 272.2 & 3.7 & 66.0 & 0.06 & 12.1
    & \multicolumn{2}{c}{\cellcolor[gray]{0.9}} & 3.96 & 66.15 & 0.12 & 487.08 & 0.12 & 13.98 \\
    1000 & 2999 & 0.3 & 531.7 & 104.6 & 116.1 & 0.2 & 13.4
    & \multicolumn{2}{c}{\cellcolor[gray]{0.9}} & 118.02 & 116.53 & 0.23 & 496.66 & 0.47 & 13.98 \\
    5000 & 14999 & 19.5 & 2540 & 506.0 & 234.1 & 5.8 & 65.9
    & \multicolumn{2}{c}{\cellcolor[gray]{0.9}} & 520.05 & 234.93 & 1.21 & 592.82 &  10.27 & 17.39 \\
    10000 & 29999 & 138.7 & 5051 & 394.3 & 396.4 & 24.2 & 98.0
    & \multicolumn{2}{c}{\cellcolor[gray]{0.9}\multirow{-5}*{Error}} 
    & 367.66 & 398.11 & 2.37 & 717.26 & 42.74 & 28.49 \\
    \bottomrule
  \end{tabular}
\end{table}

\begin{table}[htbp!]
  \centering
  \fontsize{6pt}{7}\selectfont
  \caption{Quantum circuit simulation benchmarks on GHZ circuits.
    Time is measured in seconds (s), memory usage is meausred in
    MB.}\label{tab:GHZresults}
  \begin{tabular}{ccrrrrrrrrrrrrrr}
    \toprule
    \multirow{2}*{Qubits} & \multirow{2}*{Gates}
    & \multicolumn{6}{c}{Zero to Zero Amplitude}
    & \multicolumn{8}{c}{Single Output Sample} \\
    \cmidrule(lr){3-8} \cmidrule(lr){9-16}
    & & \multicolumn{2}{c}{DDSIM} & \multicolumn{2}{c}{SliQSim} & \multicolumn{2}{c}{Ours}
    & \multicolumn{2}{c}{DDSIM} & \multicolumn{2}{c}{SliQSim} & \multicolumn{2}{c}{WCFLOBDD} & \multicolumn{2}{c}{Ours}\\
    \midrule
    & & Time & Mem & Time & Mem & Time & Mem & Time & Mem & Time & Mem & Time & Mem & Time & Mem\\[2pt]
    100 & 100 & 0.0019 & 71.2 & 0.01 & 12.5 & 0.0023 & 12.2
    & 0.002 & 71.29 & 0.009 & 14.19 & 0.03 & 478.84 & 0.005 & 13.98 \\
    500 & 500 & 0.04 & 271.8 & 0.07 & 23.9 & 0.01 & 12.3
    & 0.04 & 271.72 & 0.06 & 24.36 & 0.04 & 479.87 & 0.05 & 13.98 \\
    1000 & 1000 & 0.3 & 530.3 & 0.9 & 48.2 & 0.04 & 12.8
    & 0.26 & 530.34 & 0.91 & 48.93 & 0.06 & 481.52 & 0.18 & 13.98 \\
    5000 & 5000 & 19.9 & 2535 & 17.2 & 192.4 & 1.4 & 64.5
    & 19.58 & 2535.00 & 18.73 & 193.14 & 0.22 & 495.71 & 4.14 & 13.98 \\
    10000 & 10000 & 142.1 & 5041 & 74.9 & 324.6 & 6.2 & 66.7
    & 144.98 & 5041.17 & 78.12 & 326.00 & 0.43 & 512.53 & 16.68 & 18.71 \\
    \bottomrule
  \end{tabular}
\end{table}

As shown in \cref{tab:BVresults,tab:GHZresults}, FeynmanDD performed better than
DDSIM and SliQSim in terms of both time and memory usage.
The results confirm that FeynmanDD is capable of efficiently simulating these
relatively simple circuits and can handle larger sizes of GHZ and BV circuits
with little additional cost.
In these tests, WCFLOBDD achieves the fastest runtime, while its memory usage is
significantly higher than that of FeynmanDD\@.

\subsection{Linear Network Circuits}\label{sec:linear_network_exp}

These are IQP circuits based on degree-$3$ polynomials~\cite{Mon17}
$f: \bin^{n} \to \bin$ as follows:
\begin{equation}\label{eq:small_bdd_f}
  f(x) = A(x) \sum_{i=1}^n x_i + \sum_{i=1}^{n-k+1} C_{i:i+k-1},
\end{equation}
where $A(x) := \sum_{i=1}^n \alpha_i x_i$, with $\alpha_i$ randomly selected
from $\bin$ and $k = \order{\log{n}}$ $C_{i:j}$ consists exclusively of
degree-$3$ terms involving variables $x_i, \ldots, x_j$.
The design purposefully ensures that the second term $C_{i:i+k-1}$ involves $k$
consecutive variables, guaranteeing that the number of forward signals for each
module in Figure 23 of~\cite{Knu09} remains bounded by $k+1$.
An important reason to include the $C_{i:i+k-1}$ terms in this construction is
that they are implemented using CCZ gates, which are non-Clifford and prevent the
use of the Gottesman-Knill algorithm~\cite{Got98,AG04}. A more comprehensive analysis 
of this circuit family will be presented in subsequent work.
For each $(n, k)$ pair, ten circuits were generated.
The values (including the number of gates) in \cref{tab:linearnetworkresults}
are the average values.

\Cref{tab:linearnetworkresults} demonstrates FeynmanDD's strength in simulating
linear-network circuits.
For fixed $n$, increasing $k$ moderately increased FeynmanDD's runtime, but
still much faster than DDSIM\@.
For instance, when $n=30$, $k=7$ in amplitude task, FeynmanDD took $0.003$s and
$12$MB, while DDSIM required approximately $1065$s and $662$MB on average.
For the simulation task, FeynmanDD took $1$s and $26$MB while DDSIM used $1428$s
and $722$MB on average.
In contrast, WCFLOBDD was unable to complete the computation within one hour.

\begin{table}[htbp!]
  \centering
  \fontsize{6pt}{7}\selectfont
  \caption{Quantum circuit simulation benchmarks on linear-network circuits.
    In the table $n$ is the number of qubits, $k$ is a constant measuring gate
    locality.
    Time is measured in seconds (s), memory usage is meausred in MB\@.
    SliQSim is not compared as it currently does not support $\CCZ$
    gates.}\label{tab:linearnetworkresults}
  \begin{tabular}{cccrrrrrrrrrr}
    \toprule
    \multirow{2}*{$n$} & \multirow{2}*{$k$} & \multirow{2}*{Gates}
    & \multicolumn{4}{c}{Zero to Zero Amplitude}
    & \multicolumn{6}{c}{Single Output Sample} \\
    \cmidrule(lr){4-7} \cmidrule(lr){8-13}
    & & & \multicolumn{2}{c}{DDSIM} & \multicolumn{2}{c}{Ours}
    & \multicolumn{2}{c}{DDSIM} & \multicolumn{2}{c}{WCFLOBDD} & \multicolumn{2}{c}{Ours}\\
    \midrule
    & & & Time & Mem & Time & Mem & Time & Mem & Time & Mem & Time & Mem \\[2pt]
    20 & 5 & 162.0 & 0.20 & 34.17 & 0.002 & 12.15 & 0.19 & 34.21 & 187.06 & 12466.84 & 0.14 & 13.79\\
    20 & 7 & 163.0 & 0.30 & 38.20 & 0.003 & 12.13 & 0.29 & 38.23 & 464.12 & 21069.78 & 0.64 & 16.74\\
    30 & 5 & 319.8 & 459.07 {\color{red} \tiny (1)} & 157.49  & 0.003 & 12.13
        & 460.53 {\color{red} \tiny (1)} & 157.56 & TO & & 0.34 & 13.92\\
    30 & 7 & 321.1 & 1065.42 {\color{red} \tiny (4)} & 662.17 & 0.003 & 12.14
        & 1428.32 {\color{red} \tiny (3)} & 721.57 & TO & & 1.00 & 26.06\\
    40 & 5 & 527.3 & TO &  & 0.004 & 12.17 & TO & & TO & & 0.75 & 26.63\\
    40 & 7 & 531.7 & TO &  & 0.005 & 12.13 & TO & & TO & & 1.60 & 35.47\\
    \bottomrule
  \end{tabular}
\end{table}

\section{Circuit Equivalence Checking Experiments and Results}%
\label{sec:equivalence_check_results}

We evaluated the performance of FeynmanDD for the task of equivalence checking
by comparing it with MQT-QCEC (https://github.com/cda-tum/qcec), a widely used
tool that integrates multiple equivalence checking techniques.
To conduct the tests, both input circuits needed to use discrete gate sets
supported by FeynmanDD\@.
Specifically, we selected (1) a subset of quantum circuits from
RevLib~\cite{WGT+08} that FeynmanDD can process and (2) some quantum circuits
used in \cref{sec:exp}.
The experimental setup followed the same configuration as described in
\cref{sec:exp}.
The time limit for each check was set to 600 seconds (10 minutes).

\begin{table}[htb!]
  \centering
  \fontsize{6pt}{7}\selectfont
  \caption{Quantum circuit equivalence check benchmarks.
    $n$ stands for the number of qubits, $m$ is the number of gates of original
    circuit, $m'$ is the number of gates of transformed circuit, time is in
    seconds (s), TO means timeout.}\label{tab:equivalence_check_results}
  \begin{tabular}{lrrrrrrrrr}
    \toprule
    \multirow{2}*{Circuit} & \multirow{2}*{$n$} & \multirow{2}*{$m$} & \multirow{2}*{$m'$}
    & \multicolumn{2}{c}{Equivalent} & \multicolumn{2}{c}{Missing} & \multicolumn{2}{c}{Reverse} \\
    \cmidrule(lr){5-6} \cmidrule(lr){7-8} \cmidrule(lr){9-10}
    & & & & MQT-QCEC & Ours & MQT-QCEC & Ours & MQT-QCEC & Ours \\
    \midrule
    0410184\_169 & 14 & 46 & 46 & 1.73 & 0.08 & 1.51 & 0.32 & 1.25 & 0.43\\
    4gt11-v1\_85 & 5 & 8 & 8 & 0.45 & 0.02 & 0.32 & 0.02 & 0.45 & 0.02\\
    alu-v0\_27 & 5 & 11 & 11 & 0.49 & 0.02 & 0.48 & 0.01 & 0.46 & 0.01\\
    alu-v1\_29 & 5 & 12 & 12 & 0.54 & 0.02 & 0.44 & 0.02 & 0.37 & 0.003\\
    alu-v2\_33 & 5 & 12 & 12 & 0.08 & 0.002 & 0.06 & 0.002 & 0.06 & 0.002\\
    alu-v3\_35 & 5 & 12 & 12 & 0.08 & 0.002 & 0.07 & 0.002 & 0.07 & 0.002\\
    alu-v4\_37 & 5 & 12 & 12 & 0.18 & 0.002 & 0.16 & 0.003 & 0.12 & 0.003\\
    apex2\_289 & 498 & 1785 & 1779 & 38.13 & TO & 36.36 & TO & 137.44 & TO\\
    avg16\_324 & 576 & 3996 & 3996 & 57.39 & TO & TO & TO & TO & TO\\
    avg8\_325 & 320 & 2013 & 2013 & 27.39 & TO & TO & TO & TO & TO\\
    bw\_291 & 87 & 312 & 312 & 0.90 & 2.33 & 0.96 & 3.39 & TO & 6.80\\
    c2\_181 & 35 & 116 & 116 & 2.87 & TO & 0.56 & 81.36 & 2.60 & TO\\
    cps\_292 & 923 & 2787 & 2779 & 35.76 & TO & 27.50 & TO & TO & TO\\
    cycle10\_293 & 39 & 90 & 88 & 3.13 & 0.03 & 3.21 & 0.02 & 5.04 & 0.06\\
    e64-bdd\_295 & 195 & 452 & 452 & 13.78 & 0.41 & 28.31 & 0.20 & 13.22 & 0.18\\
    ham7\_106 & 7 & 32 & 32 & 1.92 & 4.15 & 2.28 & 2.53 & 1.80 & 4.02\\
    ham7\_299 & 21 & 68 & 62 & 0.59 & 0.01 & 0.77 & 0.005 & 0.69 & 0.005\\
    pdc\_307 & 619 & 2096 & 2096 & 11.52 & TO & 11.60 & TO & 412.45 & TO\\
    spla\_315 & 489 & 1725 & 1725 & 34.94 & TO & 21.39 & TO & 27.96 & TO\\
    sym6\_316 & 14 & 35 & 35 & 1.08 & 0.02 & 0.98 & 0.03 & 0.98 & 0.03\\
    sym9\_317 & 27 & 71 & 71 & 3.05 & 0.12 & 1.87 & 0.27 & 3.17 & 0.38\\
    \midrule
    GHZ\_100 & 100 & 100 & 100 & 4.88 & 0.01 & 4.96 & 0.01 & 5.00 & 0.01\\
    GHZ\_500 & 500 & 500 & 500 & 13.22 & 0.35 & 13.47 & 0.32 & 13.29 & 0.47\\
    GHZ\_1000 & 1000 & 1000 & 1000 & 15.16 & 1.35 & 15.04 & 1.49 & 15.29 & 1.21\\
    GHZ\_5000 & 5000 & 5000 & 5000 & 311.40 & 26.10 & 354.13 & 26.36 & 320.29 & 25.71\\
    GHZ\_10000 & 10000 & 10000 & 10000 & TO & 91.53 & TO & 91.72 & TO & 89.18\\
    BV100 & 100 & 299 & 299 & 4.03 & 0.04 & 4.22 & 0.08 & 4.58 & 0.04\\
    BV500 & 500 & 1499 & 1499 & 13.27 & 0.79 & 13.57 & 0.88 & 13.88 & 0.78\\
    \midrule
    linear\_20\_5\_1\_0 & 20 & 208 & 184 & 1.96 & 0.01 & 1.76 & 0.01 & \multicolumn{2}{c}{\cellcolor[gray]{0.9}} \\
    linear\_20\_7\_1\_0 & 20 & 205 & 181 & 1.99 & 0.005 & 3.71 & 0.004 & \multicolumn{2}{c}{\cellcolor[gray]{0.9}} \\
    linear\_30\_5\_1\_0 & 30 & 380 & 346 & 2.58 & 0.03 & 2.76 & 0.02 & \multicolumn{2}{c}{\cellcolor[gray]{0.9}} \\
    linear\_30\_7\_1\_0 & 30 & 407 & 369 & 2.05 & 0.01 & 5.15 & 0.03 & \multicolumn{2}{c}{\cellcolor[gray]{0.9}} \\
    linear\_40\_5\_1\_0 & 40 & 645 & 601 & 2.48 & 0.02 & 2.01 & 0.02 & \multicolumn{2}{c}{\cellcolor[gray]{0.9}} \\
    linear\_40\_7\_1\_0 & 40 & 596 & 540 & 3.39 & 0.01 & 4.06 & 0.01 & \multicolumn{2}{c}{\cellcolor[gray]{0.9}\multirow{-6}*{No CNOT}} \\
    \midrule
    cz\_v2/4x4\_10 & 16 & 115 & 181 & 1.24 & 0.29 & 1.34 & 0.42  & \multicolumn{2}{c}{\cellcolor[gray]{0.9}} \\
    cz\_v2/4x5\_10 & 20 & 145 & 229 & 0.73 & 0.93 & 0.74 & 0.91  & \multicolumn{2}{c}{\cellcolor[gray]{0.9}} \\
    cz\_v2/5x5\_10 & 25 & 184 & 289 & 0.78 & 1.01 & 0.46 & 1.34  & \multicolumn{2}{c}{\cellcolor[gray]{0.9}} \\
    cz\_v2/4x4\_5 & 16 & 61 & 100 & 0.83 & 0.005 & 0.58 & 0.01  & \multicolumn{2}{c}{\cellcolor[gray]{0.9}} \\
    cz\_v2/4x5\_5 & 20 & 79 & 132 & 0.92 & 0.01 & 1.03 & 0.01  & \multicolumn{2}{c}{\cellcolor[gray]{0.9}} \\
    cz\_v2/5x5\_5 & 25 & 99 & 164 & 1.37 & 0.01 & 1.34 & 0.01  & \multicolumn{2}{c}{\cellcolor[gray]{0.9}} \\
    cz\_v2/5x6\_5 & 30 & 120 & 199 & 1.42 & 0.04 & 1.45 & 0.06  & \multicolumn{2}{c}{\cellcolor[gray]{0.9}} \\
    cz\_v2/6x6\_5 & 36 & 144 & 238 & 1.50 & 0.24 & 1.52 & 0.22 & \multicolumn{2}{c}{\cellcolor[gray]{0.9}\multirow{-8}*{No CNOT}} \\
    \bottomrule
  \end{tabular}
\end{table}

Following convention in the literature, we considered three types of equivalence
checking tasks:
\begin{enumerate}
  \item \emph{Equivalent}: Confirming equivalence between the original and
        transformed circuits.
        We used the \texttt{transpile} function in Qiskit with an optimization
        level of \emph{O3} to generate transformed circuits.
        For Google supremacy circuits, the transformed basic gate set was
        specified as \( \{\Had, \Z, \Rz, \CZ\} \), which sometimes led to an
        increased gate count compared to the original.
  \item \emph{Missing}: Randomly removing one gate from the transformed circuit
        obtained in (1) and checking its difference with the original circuit.
  \item \emph{Reverse}: For transformed circuits containing CNOT gates after
        (1), we randomly selected one CNOT gate and swapped its control and
        target qubits, then checked for differences between the modified and
        original circuits.
        If no CNOT gates were present, this task was not performed.
\end{enumerate}

The results are summarized in \cref{tab:equivalence_check_results}, which
reports execution times for both programs.
They demonstrate that FeynmanDD generally outperforms in equivalence checking
tasks when applied to the quantum circuits from \cref{sec:exp}.
For instance, in tests with GHZ circuits, FeynmanDD runs significantly faster
than MQT-QCEC\@.
When evaluating circuits from RevLib, FeynmanDD also shows superior performance
in some cases.
However, for certain RevLib circuits, particularly those with numerous gates,
FeynmanDD's performance is less satisfactory.



\section{Summary}\label{sec:summary}

In summary, FeynmanDD represents a new method for quantum circuit analysis,
building upon the Feynman path integral concept and decision diagram-based
counting algorithms.
The approach transforms circuits with a supported gate set into specialized
tensor networks---proposed as sum-of-powers forms---effectively reducing many
quantum circuit simulation tasks to counting problems.
The method converts quantum circuits to multi-terminal BDDs and computes
quantities of interest by counting function values modulo a constant.
Experimental results demonstrate superior performance in single amplitude
computation compared to existing decision diagram methods.
FeynmanDD also delivers competitive performance in multi-qubit measurement
string sampling and equivalence checking, successfully completing certain
computations that remain intractable for other approaches.

This research reveals several promising future directions.
First, it would be interesting to explore approximation techniques to complement
the current exact computation approach, potentially enabling scaling to larger
systems.
Second, alternative methods for handling complex gates like CNOT could be
investigated, such as using linear function labels for variables instead of
introducing additional variables.
Third, further improvements might be achieved through more comprehensive circuit
simplification techniques, which were minimally utilized in our implementation.
Fourth, the implementation could expand beyond CUDD, potentially adopting Sylvan
as an alternative support package to leverage parallel computing capabilities.
Given that memory usage is not currently the bottleneck, such a transition may
yield significant performance improvements for large circuits.
Finally, this technique shows promise for broader quantum computing
applications, including quantum device validation, performance benchmarking,
quantum circuit simplification, and verifiable quantum advantage.

\section*{Acknowledgments}
The work is supported by National Key Research and Development Program of China
(Grant No.\ 2023YFA1009403), National Natural Science Foundation of China (Grant
No.\ 12347104), Beijing Natural Science Foundation (Grant No.\ Z220002),
Zhongguancun Laboratory, and Tsinghua University.
BC acknowledges the support by the National Research Foundation, Singapore, and
A*STAR under its CQT Bridging Grant and its Quantum Engineering Programme under
grant NRF2021-QEP2-02-P05.

\bibliographystyle{splncs04}

\newpage
\bibliography{sum-of-power}

\end{document}